\documentclass[journal,onecolumn]{IEEEtran}
\usepackage{cite}
\usepackage{bbm}
\usepackage{amsfonts}
\usepackage{amsthm}
\usepackage{mathtools}
\usepackage{graphicx}
\usepackage{hyperref}
\usepackage{subfigure}
\usepackage{cleveref}
\usepackage{comment}
\usepackage{nicefrac}
\graphicspath{ {./images/} }
\usepackage{ifthen}

\newtheorem{corollary}{Corollary}
\newtheorem{theorem}{Theorem}

\usepackage{amsmath, amssymb, bbm, xspace}
\usepackage{epsfig}
\usepackage{longtable}
\usepackage{color}
\usepackage{mathrsfs}
\usepackage{comment}

\usepackage{courier}

\def\bkE{{\rm I\kern-.17em E}}
\def\bk1{{\rm 1\kern-.17em l}}
\def\bkD{{\rm I\kern-.17em D}}
\def\bkR{{\rm I\kern-.17em R}}
\def\bkP{{\rm I\kern-.17em P}}

\def\bkZ{{\bf{Z}}}

\def\bkE{{\rm I\kern-.17em E}}
\def\bk1{{\rm 1\kern-.17em l}}
\def\bkD{{\rm I\kern-.17em D}}
\def\bkR{{\rm I\kern-.17em R}}
\def\bkP{{\rm I\kern-.17em P}}

\makeatletter
\newcommand{\pushright}[1]{\ifmeasuring@#1\else\omit\hfill$\displaystyle#1$\fi\ignorespaces}
\newcommand{\pushleft}[1]{\ifmeasuring@#1\else\omit$\displaystyle#1$\hfill\fi\ignorespaces}
\makeatother

\def\bkZ{{\bf{Z}}}
\def\b12{(\beta_1,\beta_2)}

\newenvironment{example}{{\noindent \bf Example}}{\hfill $\square$\hspace{-4.5pt}\vspace{6pt}}
\newcounter{example}
\renewcommand{\theexample}{\thesection.\arabic{example}}

\newcounter{remark}
\renewcommand{\theremark}{\thesection.\arabic{remark}}

\def\t{^\top}

\def\Ebb{\mathbb{E}}
\newlength{\noteWidth}
\long\def\notes#1{\ifinner
{\tiny #1}
\else
\marginpar{\parbox[t]{\noteWidth}{\raggedright\tiny #1}}
\fi\typeout{#1}}

 \def\notes#1{\typeout{read notes: #1}} 

\newcommand{\ie}{i.e.\@\xspace} 

\newcommand{\inv}{^{-1}}

\def\Ebb{\mathbb{E}}

\def\Pbb{{\mathbb{P}}}
\def\Nbb{{\mathbb{N}}}

\def\limn{\ds \lim_{n \rightarrow \infty}}

\def\diag{\mathop{\hbox{\rm diag}}}

\def\spose#1{\hbox to 0pt{#1\hss}}

\def\text #1{\hbox{\quad#1\quad}}

\def\nthinsp{\mskip -2   mu}

\def\superstar{^{\raise 0.5pt\hbox{$\nthinsp *$}}}
\def\SUPERSTAR{^{\raise 0.5pt\hbox{$*$}}}

\def\lamstarT {\lambda^{\raise 0.5pt\hbox{$\nthinsp *$}T}}

\let\forallnew\forall
\renewcommand{\forall}{\forallnew\ }
\let\forall\forallnew

\def\ds{\displaystyle}
		\def\bkE{{\rm I\kern-.17em E}}
		\def\bk1{{\rm 1\kern-.17em l}}
		\def\bkD{{\rm I\kern-.17em D}}
		\def\bkR{{\rm I\kern-.17em R}}
		\def\bkP{{\rm I\kern-.17em P}}
		\def\bkY{{\bf \kern-.17em Y}}
		\def\bkZ{{\bf \kern-.17em Z}}
		\def\bkC{{\bf  \kern-.17em C}}

{\begin{list}{}%
         {\setlength{\leftmargin}{#1}}%
         \item[]%
}
{\end{list}}

		\def\bsp{\begin{split}}
		\def\beq{\begin{eqnarray}}
		\def\bal{\begin{align*}}
		\def\bc{\begin{center}}
		\def\be{\begin{enumerate}}
		\def\bi{\begin{itemize}}
		\def\bs{\begin{small}}
		\def\bS{\begin{slide}}
		\def\ec{\end{center}}
		\def\ee{\end{enumerate}}
		\def\ei{\end{itemize}}
		\def\es{\end{small}}
		\def\eS{\end{slide}}
		\def\eeq{\end{eqnarray}}
		\def\eal{\end{align*}}
		\def\esp{\end{split}}
		\def\qed{ \vrule height7.5pt width7.5pt depth0pt}  

	\def\cp2problem#1#2#3#4{\fbox
		 {\begin{tabular*}{0.9\textwidth}
			{@{}l@{\extracolsep{\fill}}l@{\extracolsep{6pt}}l@{\extracolsep{\fill}}c@{}}
				#1 & & $#4 $ 
			\end{tabular*}}}

		\def\bkE{{\rm I\kern-.17em E}}
		\def\bk1{{\rm 1\kern-.17em l}}
		\def\bkD{{\rm I\kern-.17em D}}
		\def\bkR{{\rm I\kern-.17em R}}
		\def\bkP{{\rm I\kern-.17em P}}
		
		\def\bkZ{{\bf{Z}}}

\newcommand {\beeq}[1]{\begin{equation}\label{#1}}
\newcommand {\eeeq}{\end{equation}}
\newcommand {\bea}{\begin{eqnarray}}
\newcommand {\eea}{\end{eqnarray}}

\def\texitem#1{\par\smallskip\noindent\hangindent 25pt
               \hbox to 25pt {\hss #1 ~}\ignorespaces}

\def\bsp{\begin{split}}
		\def\beq{\begin{eqnarray}}
		\def\bal{\begin{align*}}
		\def\bc{\begin{center}}
		\def\be{\begin{enumerate}}
		\def\bi{\begin{itemize}}
		\def\bs{\begin{small}}
		\def\bS{\begin{slide}}
		\def\ec{\end{center}}
		\def\ee{\end{enumerate}}
		\def\ei{\end{itemize}}
		\def\es{\end{small}}
		\def\eS{\end{slide}}
		\def\eeq{\end{eqnarray}}
		\def\eal{\end{align*}}
		\def\esp{\end{split}}
		\def\qed{ \vrule height7.5pt width7.5pt depth0pt}  

\usepackage{amsmath, amssymb, xspace}
\usepackage{epsfig}
\usepackage{longtable}
\usepackage{color}
\usepackage{mathrsfs}

\def\LOLP{\mathsf{LOLP}}

\newboolean{showcomments}
\setboolean{showcomments}{true}
\newcommand{\jk}[1]{  \ifthenelse{\boolean{showcomments}}
{ \textcolor{red}{(JK says:  #1)}} {}  }
\newcommand{\aak}[1]{  \ifthenelse{\boolean{showcomments}}
{ \textcolor{blue}{(AAK says:  #1)}} {}  }
\newcommand{\vivek}[1]{  \ifthenelse{\boolean{showcomments}}
{ \textcolor{red}{(Vivek says:  #1)}} {}  }

\newcommand{\ignore}[1]{}
\newcommand{\ra}{\rightarrow}

\newcommand{\da}{\downarrow}
\newcommand{\prob}[1]{\mathbb{P}\left[#1\right]}
\newcommand{\Exp}[1]{\mathbb{E}\left[#1\right]}

\newcommand{\rev}{\scriptstyle{R}}
\newcommand{\nn}{\nonumber}

\begin{document}

\title{Statistical Economies of Scale in Battery Sharing}
%
%
\author{Vivek Deulkar,\ Jayakrishnan Nair and Ankur A. Kulkarni%
  \thanks{Vivek and Jayakrishnan are with the Department of Electrical
    Engineering, Indian Institute of Technology Bombay.}%
  \thanks{Ankur is with the Systems and Control Engineering group,
    Indian Institute of Technology Bombay.}%
  }
\maketitle

\begin{abstract}
  The goal of this paper is to shed light on the statistical economies
  of scale achievable from sharing of storage between renewable
  generators. We conduct an extensive study using real world wind data
  from a grid of equispaced wind generators sharing a common
  battery. We assume each generator is contracted to meet a certain
  demand profile to a prescribed level of reliability. We find that
  the statistical diversity in wind generation across different
  locations yields useful economies of scale once the grid spacing
  exceeds 200~km. When the grid spacing exceeds 500~km, we find that
  the economies grow dramatically: The shared battery size becomes
  insensitive to the number of participating generators. This means
  that the generators can access a common, shared battery and
  collectively achieve the same reliability they would have, had each
  of them had the \textit{entire} battery to themselves. To provide a
  rigourous foundation for this remarkable observation, we propose a
  mathematical model that demonstrates this phenomenon, assuming that
  the net generation (generation minus demand) processes associated
  with the generators are \emph{statistically independent}. The result
  is derived by characterizing the large deviations exponent of the
  loss of load probability with increasing battery size, and showing
  that this exponent is invariant with the number of generators.

\end{abstract}

\section{Introduction}

Dire concerns about the environment, volatility of prices and the
finiteness of natural resources have forged a global resolve for a
transition to more sustainable sources of energy such as wind and
solar energy. A fundamental challenge in the large-scale adoption of
such sources is their uncertainty and intermittency. Unlike
conventional generation, wind and solar energy is not available
``on-demand''. It is a function of the weather, making the
power generated random and time varying. One way to mitigate this
problem is to bundle renewable generation with a storage device, such
as a battery. This allows surplus generation to be stored by charging
the battery which can be discharged in times of deficit.

Despite many improvements on the technology front, a major roadblock
in battery adoption so far has been the cost. A consumer is usually
concerned about reliability, and the size of the battery required for
reliability comparable to that of conventional generation is extremely
high, making such batteries unaffordable by individual generators
\cite{Deulkar19,Dunn2011,EIA18}. Community storage is often considered
as a solution to this predicament, but it does not quite solve the
problem because, owing to strong correlations in generation (say
between various solar photovoltaic installations in the community) and
consumption within the community, as the number of users in the
community increases, the battery size requirements also scale
proportionately.

Our main thrust in the present paper is in showing that there is an
\textit{economy of scale} that kicks in when a battery is
\textit{shared} amongst \textit{suitably located} users. These users
can access a common, shared battery and collectively achieve the same
reliability they would have, had each of them had the \textit{entire}
battery to themselves. This allows the cost of the battery to be
divided in principle by \textit{any number} of users \textit{without}
sacrificing reliability, bringing down dramatically the cost of
battery ownership. Shedding light on this remarkable economy of scale
is the goal of this paper.

Economy of scale traditionally refers to the reduction in the unit
cost of production as an enterprise increases its output.  Production
of a good involves fixed costs and variable costs. Fixed costs
includes costs such as the cost of acquiring land, purchasing the
production equipment and installation costs. Variable costs are input
costs associated with production including raw material, energy, human
resources and so on. Thus when larger quantities are produced, the
fixed cost is distributed over a larger number of units, thereby
reducing the total cost per unit of the good.
In this paper we demonstrate a novel kind of economy of scale. Unlike
the classical notions, the basis of this economy of scale is the
\textit{statistical} diversity in renewable generation from sites
sufficiently geographically separated. This diversity makes it
possible for an instantaneous energy deficit at any location to be
satisfied using the instantaneous surplus at other locations some of
the time, diminshing the reliance on battery storage.

We demonstrate this economy of scale via an extensive data study using
past wind data obtained from NREL\cite{url_wind_data_windtoolkit}. We
consider grids of roughly equispaced wind generators sharing a common
battery, with the goal of meeting a constant local power demand
with high reliability. This demand may be interpreted as the
contracted commitment of the generator to the electricity market,
which it is required to fulfil to a prescribed level of
reliability\cite{Bitar12}. For each subset of generators, we compute
the shared battery size required to meet the demand, given a
prescribed bound on the loss of load probability ($\LOLP$), which is
the long run fraction of time the generator is unable to cater its
demand. Assuming the wind generators to be of comparable scale, a
sub-linear growth in the shared battery size with respect to the
number of generators $N$ would indicate economies of scale via battery
sharing. We find that these economies kick in once the grid spacing
exceeds 200~km. Once the grid spacing exceeds 500~km, the economies
become truly remarkable; in this case, we find that the shared battery
requirement remains insensitive to the number of participating
generators $N.$ This means that the battery size required to meet a
high degree of reliability at \textit{any single} location would suffice to
meet the joint requirement at $N$ locations, allowing for an $N$-fold
reduction in the cost of battery ownership. We also find that when the
$\LOLP$ targets are stricter (i.e., loss of load is a \emph{rare
  event}), the grid spacing needs to be larger for economies of scale
to be observed. This is perhaps because rare events in weather (such
as storms and depressions) are closely correlated across larger
distances, even while non-rare events are not. Interestingly, for a
grid spacing exceeding 750~km, the battery requirement shows a
\textit{decrease} with $N,$ suggesting a negative correlation between
wind power generation across very long distances.

To provide a theoretical foundation for the observed economies of
scale, we propose a mathematical model, where a battery is shared
among $N$ \textit{prosumers} -- users that generate \textit{and}
consume energy -- endowed with unreliable net generation. These users
could be standalone microgrids equipped with a stochastic source of
generation, or individual homes or industries that are equipped with a
solar PV installation. The net generation (\ie generation minus
demand) of user $i$ is a stochastic process denoted $X_i(\cdot).$ We
assume $\Ebb[X_i(\cdot)]$ is positive in steady-state and that the
processes $X_1(\cdot),\hdots,X_N(\cdot)$ are statistically
independent.  Users would like to minimize the \textit{loss of load
  probability} ($\LOLP$), \ie, the steady state probability that their
energy demand is unmet. With this goal, we assume the battery is
charged by a greedy policy -- when there is surplus generation, the
users charge their battery at the surplus rate and when there is
deficit, they discharge it at the deficit rate. In the case of a
single user $i$, the loss of load probability decays exponentially
with the size of the battery. Thus for a loss of load probability of
$\epsilon$, the battery size needed is given by $b_i \approx
\frac{1}{\lambda_i}\log \frac{1}{\epsilon}$ where $\lambda_i$ is the
exponent in the above exponential decay. We prove that if these $N$
users have independent net generation and share their battery, the
resulting exponent $\lambda$ of the combined system, is
\textit{invariant} in $N$ and is upper bounded by $\min_i
\lambda_i$. This implies that if a battery of size $B$ suffices for a
LOLP of $\epsilon$ for each of the $N$ independent users
\textit{individually}, then a \textit{single} battery of size $B$
suffices when shared between all of them collectively.

Importantly, the above result \textit{does not} require $N$ to be
large. The root of the economy of scale is \textit{not} in the sharing
of the battery among large number of users. Rather the crux of this
result is the assumption of independence of the net generation
processes of individual users. Of course, statistical independence is
an theoretical idealization that is not possible in real world weather
data. However, our data studies reveal that geographical separation
can serve as a proxy for independence in practice.
%
%

The \textit{sharing economy}, pioneered by services such as AirBnB and
Lyft, has disrupted many spheres of business by increasing affordability
of services and enhancing resource utilization. Our work shows that the
inherent statistical characteristics of wind generation provide us,
through the sharing of storage, tantalizing opportunities for
significantly scaling up the deployment and integration of wind
generation and transitioning towards a future of cleaner energy.
While sharing storage results in savings, it also incurs costs due to
the transmission infrastructure required for sharing.  Thus, our
results highlight the importance of investment in massive
continent-scale grids and inter-regional cooperation to facilitate
sharing across wider geographical regions.  The work also leads
naturally to novel questions on the optimal design of battery sharing
arrangements. Given a pool of users, which users should be come
together and share their battery? What are the optimal partitions
among this pool? How many batteries and of what size do we need? We
plan to address this in forthcoming work.

\subsection*{Related literature}
Previous studies on the power spectral density of the wind time-series
have attempted to compute auto-correlations and spatial correlations
of wind speeds \cite{apt2007spectrum,
  katzenstein2010variability,fertig2012effect}. These studies attempt
to estimate these correlations with the aim of approximating the wind
speed process in closed form. These studies have also found
correlations in the wind speeds to drop beyond the range of about
200-250 km. Of course, these studies have not concerned themselves with
the question of battery sharing.

There is considerable recent literature on energy sharing between
users having their own storage, as well as on battery sharing between
users. One line of work in this space focuses on the game theoretic
aspects of the interaction between uses, in a non-cooperative
framework (see, for example, \cite{kalathil2019sharingelec}), or in a
cooperative framework (see, for example,
\cite{chakraborty2018coalationalsharing,wu2016communitystorage}).

Another line of work focuses on the scheduling aspects of energy
storage, typically using a Markov decision process (MDP) framework
(for example, see \cite{korpaas2003opersizingstorage,
  kim2011optenergy,vandevan2012optcontenergy,
  zhou2018windstoragemdp}). These papers derive structural properties
of the optimal scheduling policies and propose useful heuristic
algorithms.

In contrast, the focus of the present paper is on statistical
economies of scale in battery sizing when the battery is shared
between multiple prosumers. The work closest to ours is
\cite{Deulkar19}, which characterizes the battery size required by a
single user in order to meet a certain reliability threshold via large
buffer asymptotics. To the best of our knowledge, there is no prior
work that characterizes the scaling of (shared) battery size 
with the number of users either analytically or through data.

\section{Data Study}
\label{sec:numerics}


The goal of this section is to empirically evaluate the economies of
scale achievable via battery sharing in practice. Specifically, we use
real world wind generation data from several locations across the
U.S. to construct a net generation time series for each location
(assuming a steady demand). We then evaluate the economies of scale in
battery sizing by computing the volume of shared storage required by
each subset of the locations in order to meet a given reliability
threshold. We find that there are nontrivial but modest economies of scale when the
locations are between 200-500 km apart, that improve as the spacing increases. Once the spacing
between the locations exceeds 500 km we find that, remarkably, the shared battery size needed in order to cater to $N$
locations \emph{does not} grow with $N.$ This implies that an $N$-fold
reduction in the cost of battery ownership/operation is achievable for
\emph{any} $N,$ limited only by the ability to dynamically transfer
power between the locations and the storage infrastructure.


\subsection{Data collection}

We gathered six years (2007--2012) of wind power generation data from
different locations in the United States (see the left panels of
Figures~\ref{fig:scenario1}--\ref{fig:scenario6}). The data was
obtained from the Wind Integration National Dataset (WIND) Toolkit
available from the National Renewable Energy Laboratory
(NREL).\cite{url_wind_data_windtoolkit} The WIND Toolkit contains the
meteorological conditions and wind turbine generation for around
1,26,000 sites in the U.S. for the years 2007--2013. These weather
observations and mesoscale climate data are further analyzed using the
Weather Research and Forecasting (WRF) model to `interpolate' an
underlying meteorological data set of much finer spatial
resolution. Thus, the wind power data available on the NREL
website~\cite{url_wind_data_windtoolkit} corresponds to a 2~km$\times$2
km grid spanning the continental United States. The interval between
succesive samples is five minutes. At each location, the generation
data corresponds to eight wind turbines, each with a rated power of 2
megawatt (MW) and a hub height of 100m.
Thus, the power generation time series corresponding to each grid
location ranges from 0--16~MW.

\begin{figure}[htp]
  
  \centering 
  \subfigure [Locations]
  {\includegraphics[width = 0.3\textwidth,trim={1cm 0 4cm 0},clip]{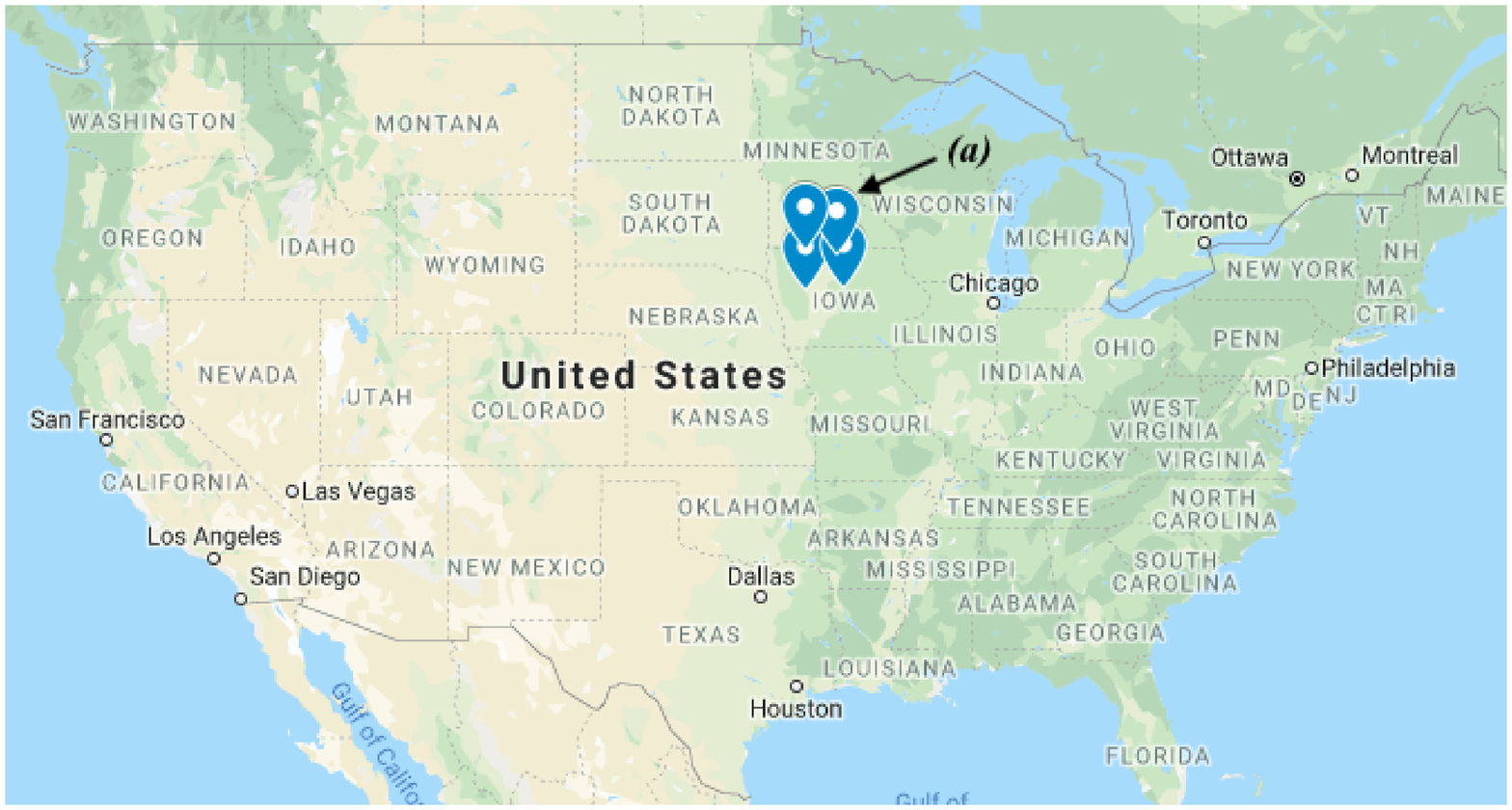} }
  \subfigure [Battery requirement]
  {\includegraphics[width = 0.3\textwidth,trim={1cm 0 0 0}]{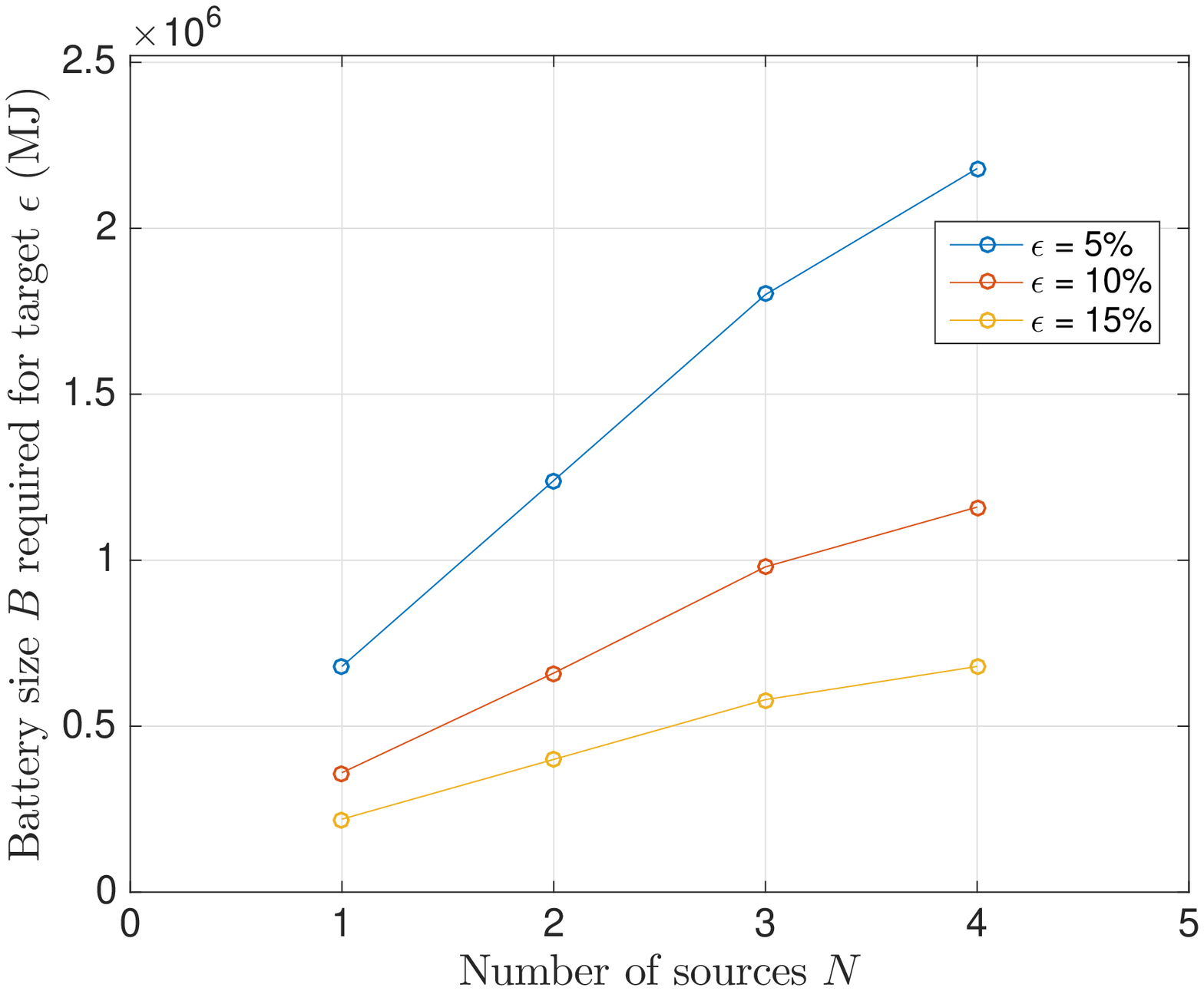} }
  \caption{Scenario 1: Locations within 200 km of one another}
  \label{fig:scenario1}

  \centering 
  \subfigure [Locations]
  {\includegraphics[width = 0.3\textwidth,trim={1cm 0 4cm 0},clip]{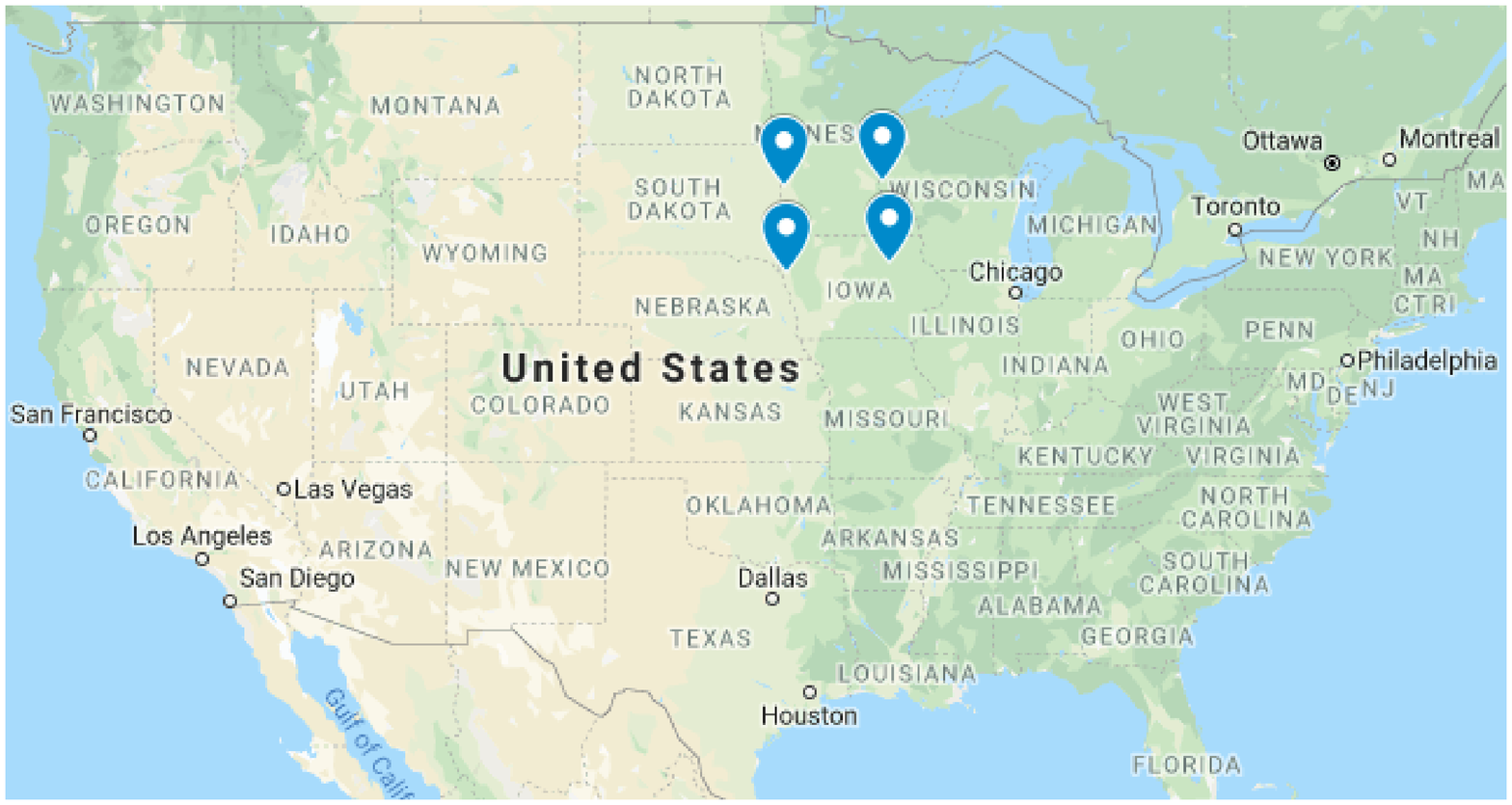} }
  \subfigure [Battery requirement]
  {\includegraphics[width = 0.3\textwidth,trim={1cm 0 0 0}]{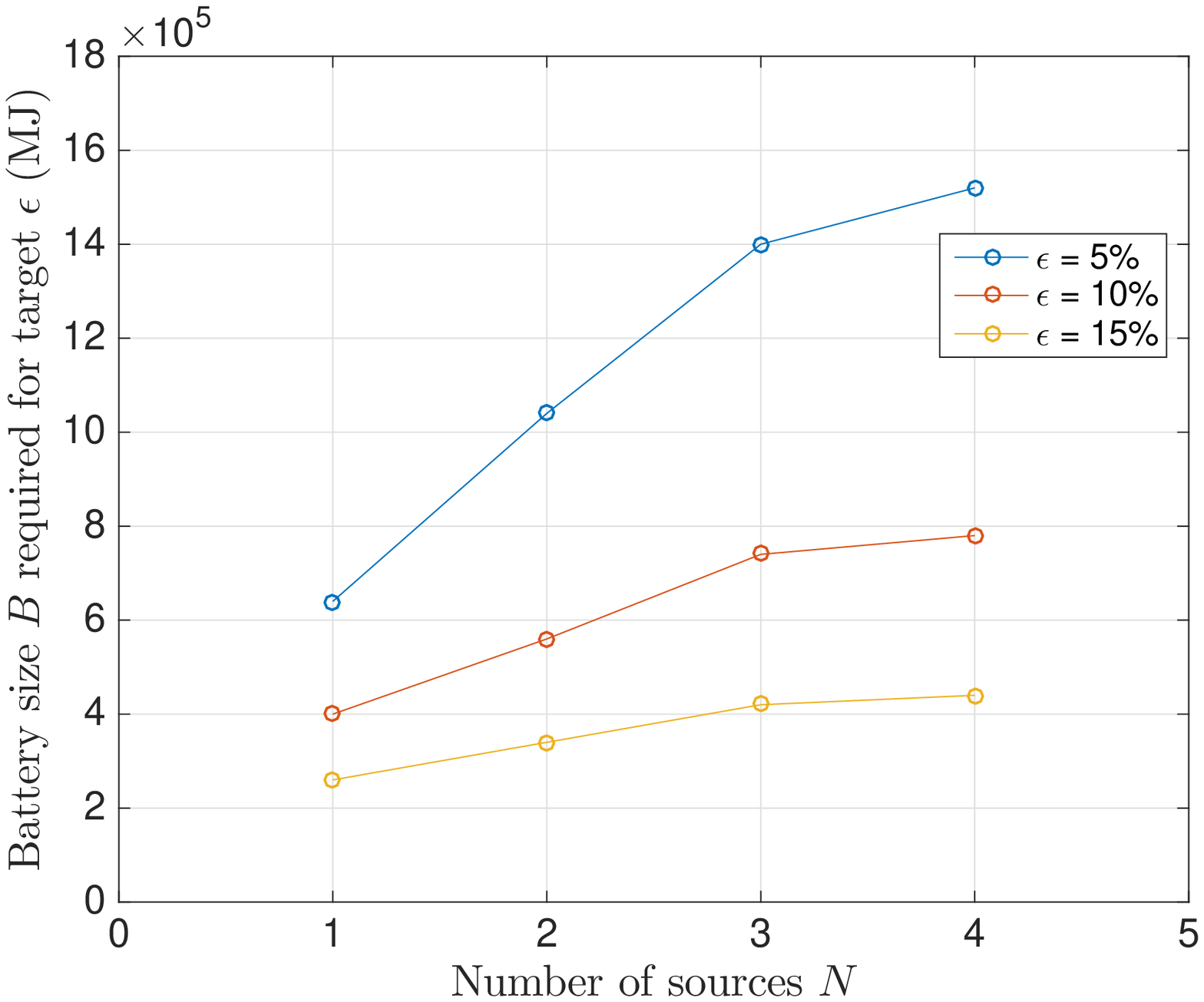} }
  \caption{Scenario 2: Locations between 200--500 km of one another}
  \label{fig:scenario2}

  \centering 
  \subfigure [Locations]
  {\includegraphics[width = 0.3\textwidth,trim={1cm 0 4cm 0},clip]{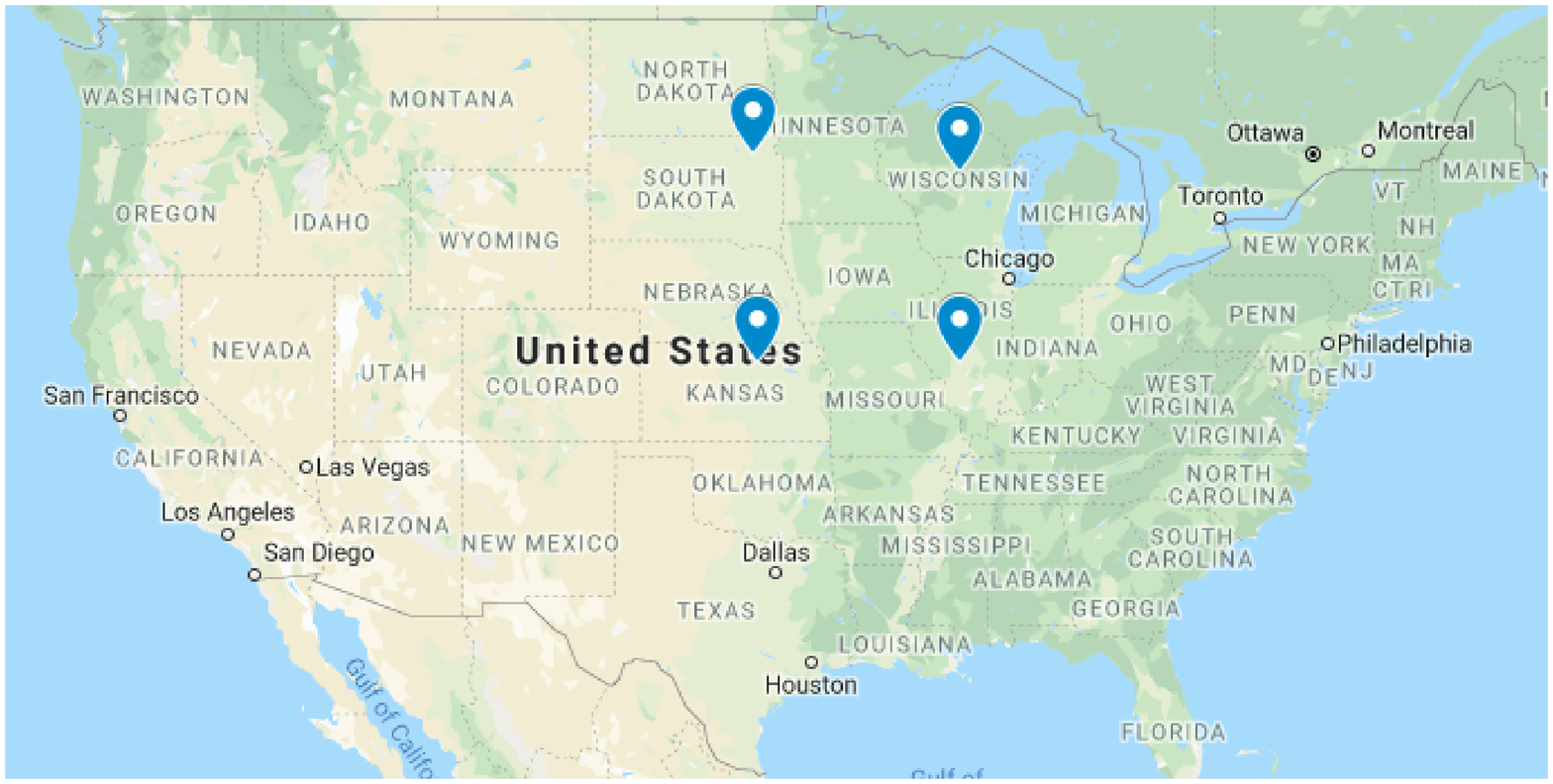} }
  \subfigure [Battery requirement]
  {\includegraphics[width = 0.3\textwidth,trim={1cm 0 0 0}]{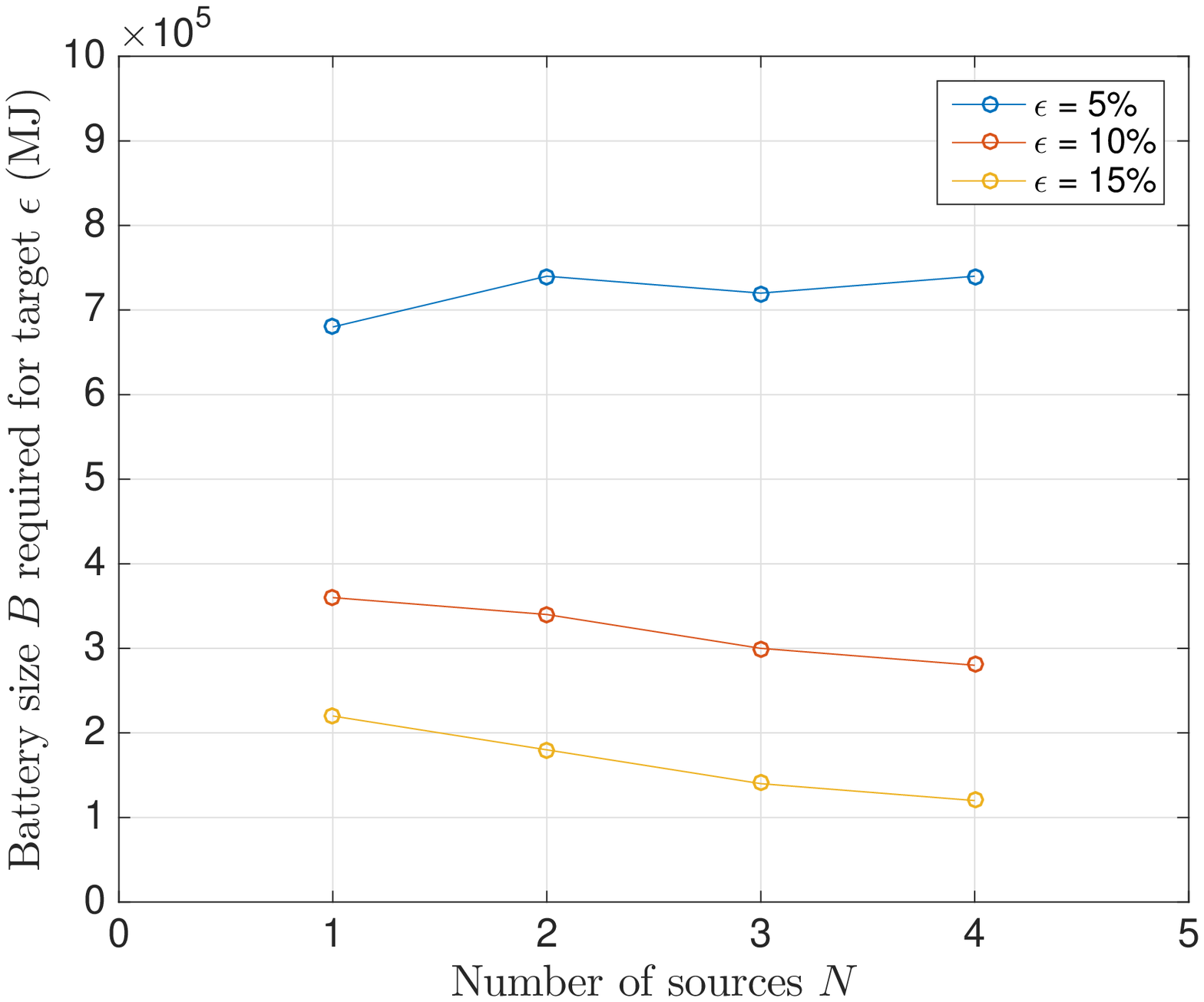} }
  \caption{Scenario 3: Locations between 500--1000 km of one another}
  \label{fig:scenario3}

  \centering
  \subfigure [Locations]
  {\includegraphics[width = 0.3\textwidth,trim={1cm 0 4cm 0},clip]{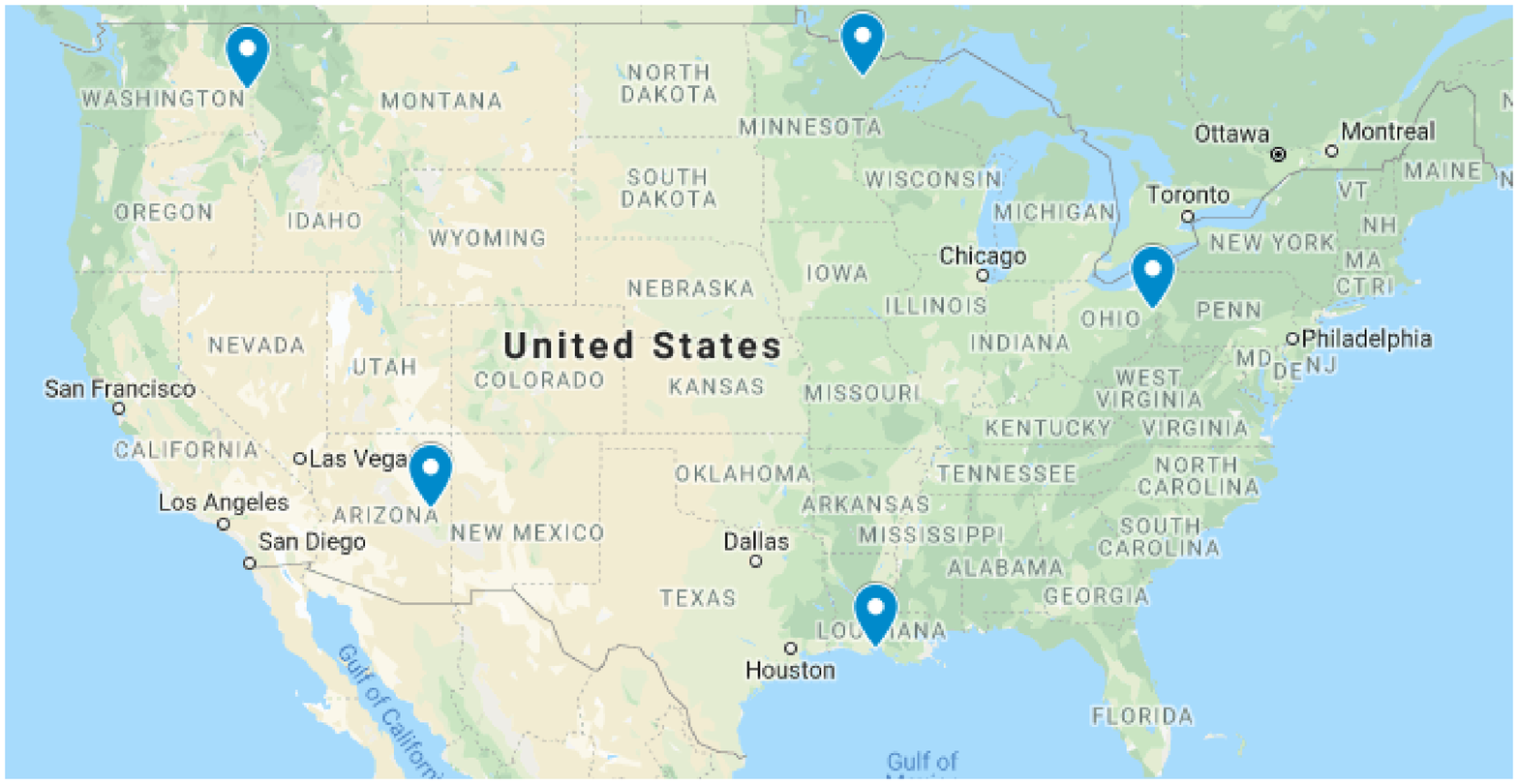} }
  \subfigure [Battery requirement] 
  {\includegraphics[width = 0.3\textwidth,trim={1cm 0 0 0}]{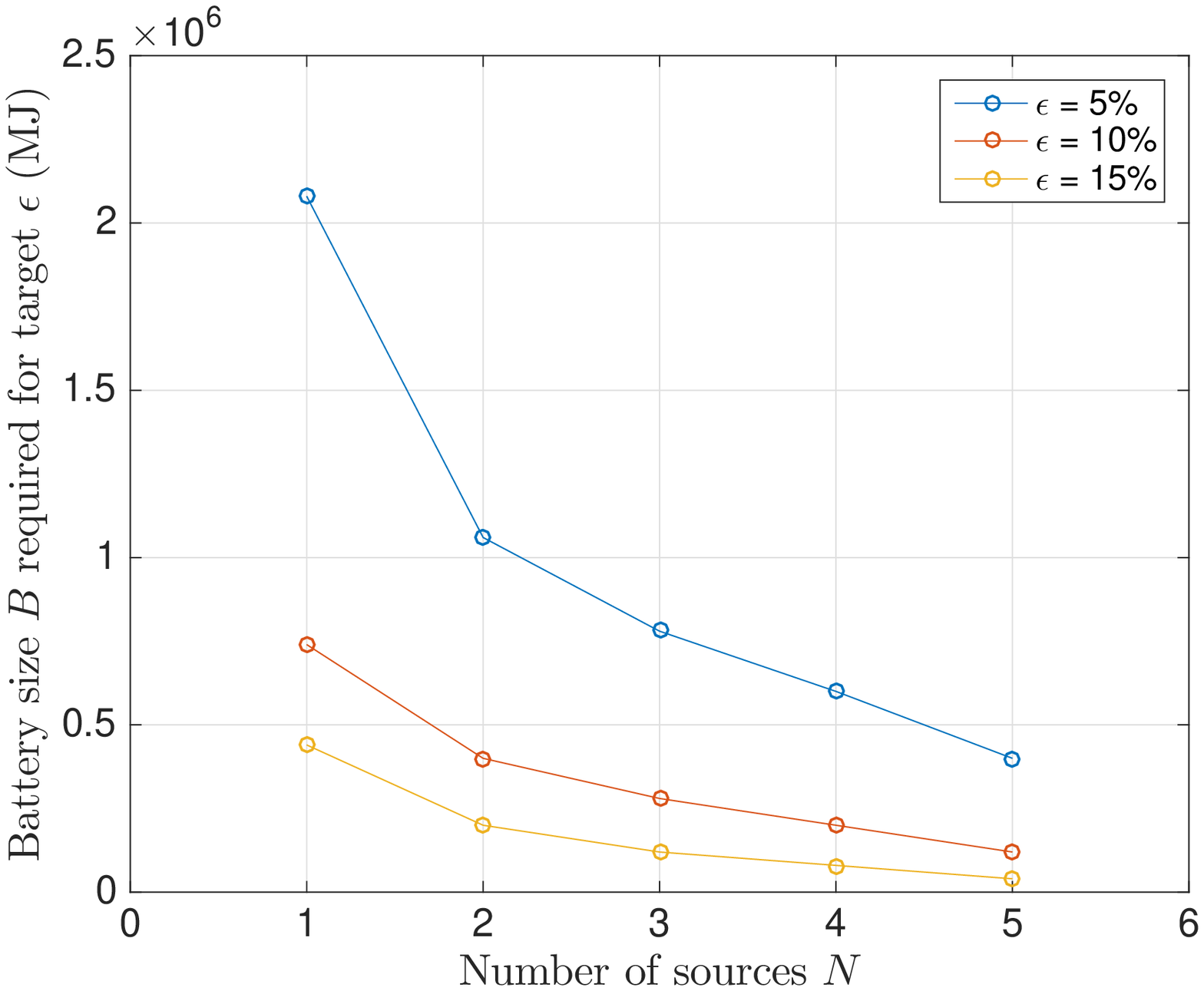} }
  \caption{Scenario 4: Locations over 1000 km from one another}
  \label{fig:scenario4}

\end{figure}

\subsection{Simulation methodology}

From the wind power generation data, we construct the net generation
trace as follows. For each location, we assume that the demand is
constant, and equal to 60\% of the (time) average generation. This
corresponds, for example, to a wind generator committing to a steady
power delivery of 60\% of its average generation to the grid, using
battery storage to smoothen the temporal supply
intermittency.\footnote{In this study, we keep the demand side of the
  net generation simple in order to focus on the impact of the
  \emph{supply size intermittency} on battery sizing. Considering more
  realistic demand traces that capture diurnal/seasonal variations
  presents an avenue for future work.}
Note that the resulting average net generation (or drift) at each
location is positive.

To get a sense of the scale of the battery requirement, consider
location~$(a),$ marked on Figure~\ref{fig:scenario1}. The average power
generation at this location equals 6.3498 MW. The battery size required
for this location alone, in order to achieve a target $\LOLP$ of 10\%,
equals 360000 MJ, which is 100 megawatt-hours.\footnote{For reference, as
  of to the best our of knowledge, as of 2018, the largest utility scale lithium-ion battery installation
  was built by Tesla for the Hornsdale Power Reserve in Australia,
  with a capacity of 129
  megawatt-hours. See \cite{url_HornsdalePowerReserve}} A battery
of this size would, when fully charged, be able to meet the demand (60\% of average generation) on
its own for 26.2475 hours! 

Given the net generation time series for each location as described
above, we compute the battery requirement corresponding to a given
subset $\mathcal{S}$ of locations as follows. For a shared battery of
size $B,$ we simulate the temporal evolution of the battery occupancy
driven by the aggregate supply/demand from all locations in
$\mathcal{S}.$ The $\LOLP$ for this scenario is then taken to be the
fraction of time the battery is empty and unable to serve the
(aggregate) demand. We then compute the minimum shared battery size
required by the locations in $\mathcal{S}$ to achieve a prescribed
$\LOLP$ target of $\epsilon$ by running the above simulation for
various battery sizes~$B.$

With the above methodology, our case study is presented as follows. We
explore the economies of scale in battery sharing between different
sets of locations, ranging from sets where all locations are less than
200 km apart, to sets where all locations are over 1000 km apart. For
each set of locations, since there are multiple subsets of any given
size, we plot the maximum battery requirement in order to meet an
$\LOLP$ target of $\epsilon,$ among all subset of the locations having
size $N$ (see the right panels of
Figures~\ref{fig:scenario1}--\ref{fig:scenario6}). Note that the
scaling of the battery requirement with $N$ gives an indication of the
economies of scale obtained via battery sharing. At one extreme, if
the required battery size grows proportionately with $N,$ then there
are no economies of scale from sharing. On the other hand, if the
battery requirement stays relatively insensitive to $N,$ that would
imply substantial economies of scale (with an $N$-fold reduction in
storage requirement thanks to sharing).




\subsection{Observed economies of scale}

We first consider locations that are close by. We selected 4
locations, placed approximately at the vertices of a square, such that
each pair of locations is at most 200~km apart, as shown in
Figure~\ref{fig:scenario1}. We found that the battery requirement
associated with the `worst case' subset of size $N$ among these
locations grows nearly proportionately with $N.$ This suggests that
the wind generation from locations this close are strongly correlated,
providing little opportunity for economy from battery sharing.

Next, we consider four locations, such that the pairwise distances are
between 200--500~km; see Figure~\ref{fig:scenario2}. Once again, we
see that the battery requirement grows with $N,$ though sub-linearly
(as opposed to the near-linear growth in the previous scenario). This
suggests that there are modest economies of scale to be obtained in
this scenario.
Interestingly, the economies of scale are more pronounced for the less
stringent $\LOLP$ targets (10\% and 15\%). We believe that this is because the
`degree of independence' between the net generation processes varies
with the `rareness' of the events that cause loss of load. When the
reliability target $\epsilon$ is small, the $\LOLP$ is dictated by
rare weather events, which might exhibit strong correlation across
large distances. On the other hand, with a milder the reliability
target $\epsilon,$ the $\LOLP$ is dictated by relatively frequent and
modest dips in the net generation process, which are (perhaps) less
correlated across large distances.

Next, we consider the scenarios depicted in
Figures~\ref{fig:scenario3} and~\ref{fig:scenario4}. The pairwise
distances between locations is between 500--1000 km in the former
scenario, and more than 1000 km in the latter. We see here
that there are battery requirement does not grow with $N$, revealing
remarkable economies of scale. In fact, when the locations are over
1000 km apart, we notice that the shared storage requirement actually
shrinks with the number of locations $N.$ While a meterological
explanation of this phenomenon is beyond the scope of this paper, our
results suggest a negative correlation between wind
generation at locations very far apart.

Finally, we consider two scenarios using a larger set of locations
placed along a grid. In the first, we consider a grid of 11 locations,
such that the adjacent locations are around 200 km apart (see
Figure~\ref{fig:scenario5}). In the second, 11 locations are chosen
such that the adjacent locations are around 500 km apart (see
Figure~\ref{fig:scenario6}).  In these scenarios, we notice that
battery requirements seem to grow with $N$ when $N$ is small
(consistent with our earlier observations), but plateau as $N$
becomes larger. This is because when $N$ is larger, all subsets of
size $N$ in the grid contain locations that are further apart compared
to the grid spacing, contributing to the economies of scale. Moreover,
note that the economies of scale are more prominent under the less
stringent $\LOLP$ targets (10\% and 15\%), also consistent with our
earlier observations.

\begin{figure}[htp]
  
  \centering \subfigure [Locations]
  {\includegraphics[width = 0.3\textwidth,trim={1cm 0 4cm 0},clip]{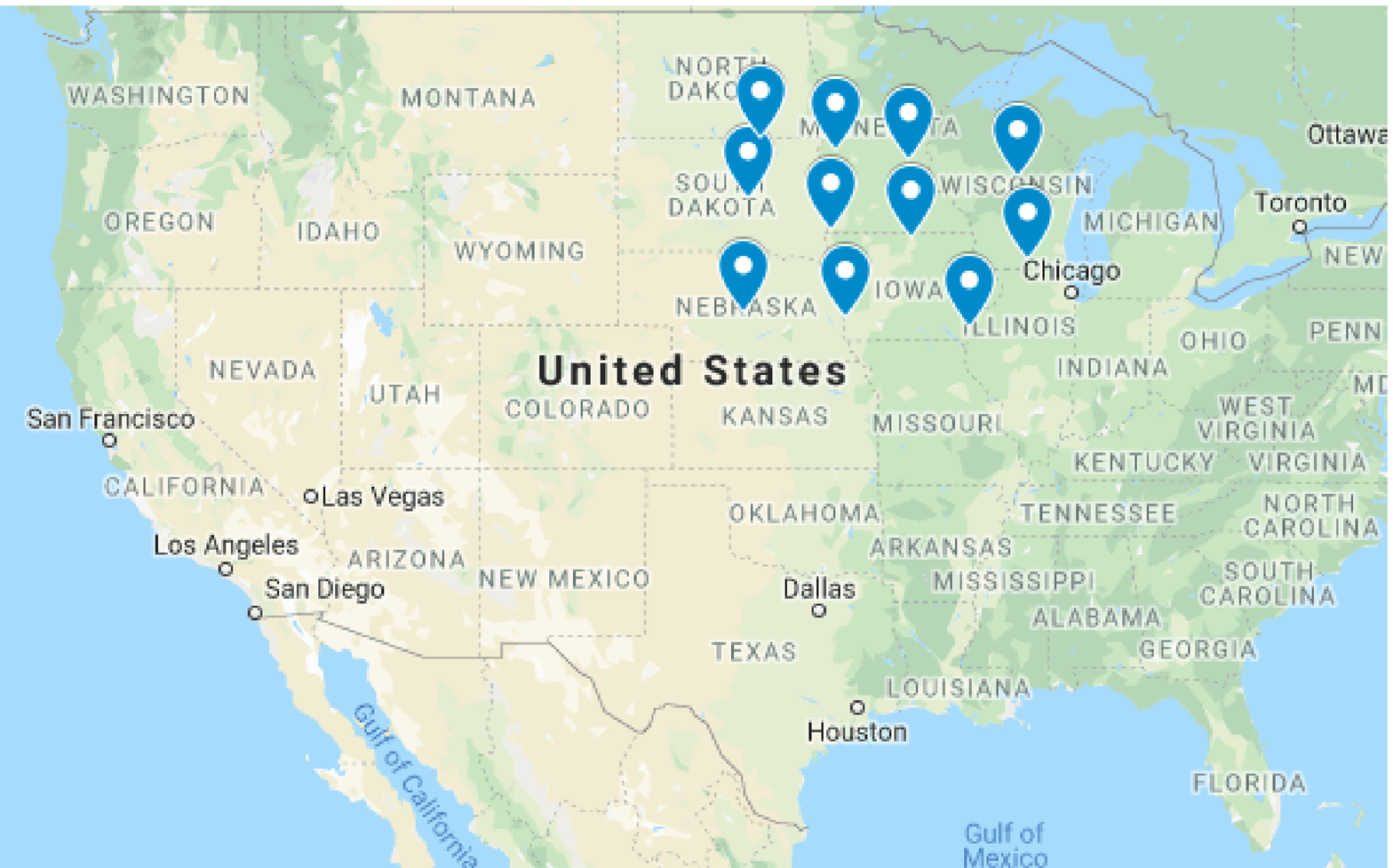} }
  \subfigure [Battery requirement]
  {\includegraphics[width = 0.3\textwidth,trim={1cm 0 0 0}]{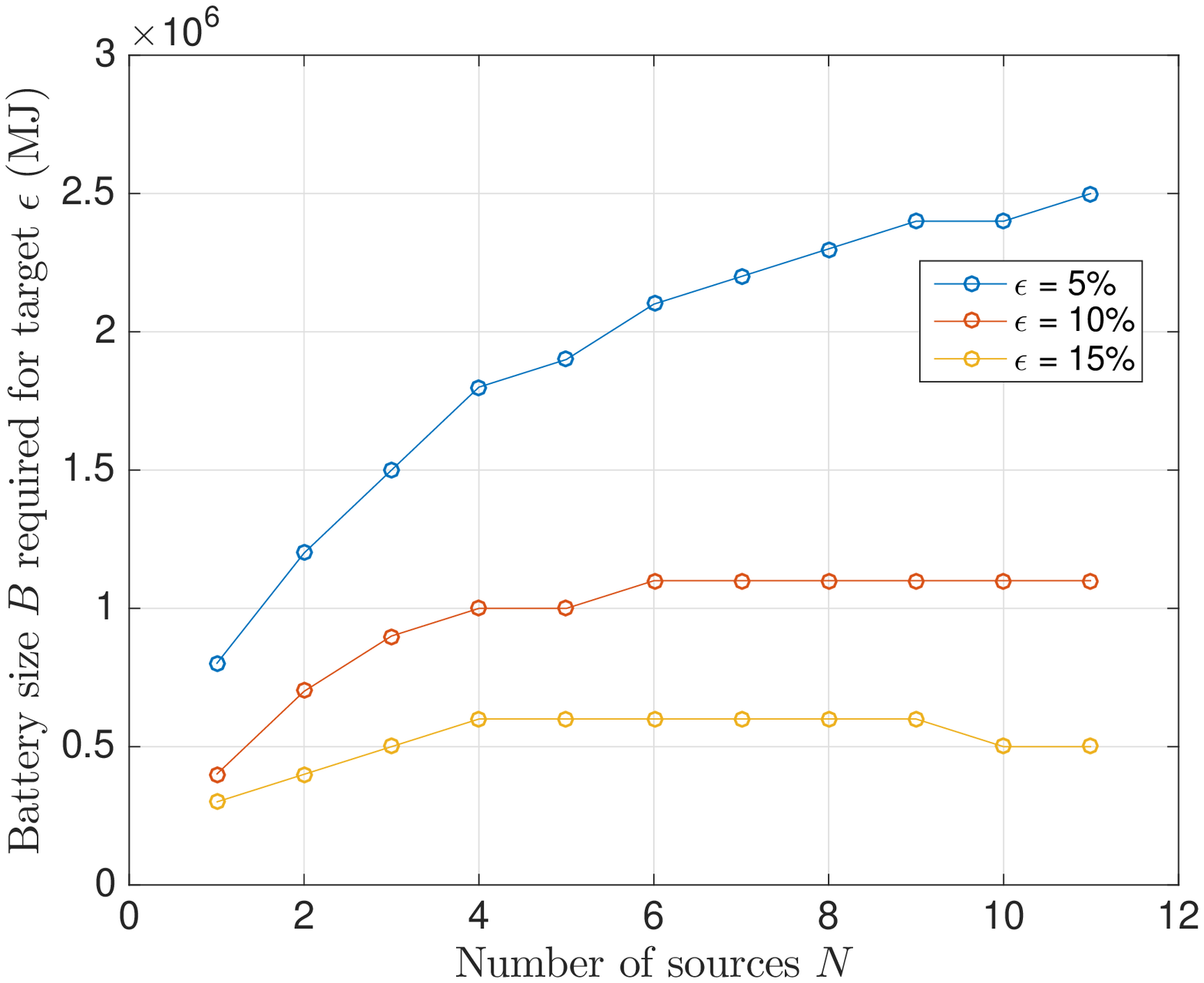} }
  \caption{Scenario 5: Grid with adjacent locations around 200 km apart}
  \label{fig:scenario5}

  \centering 
  \subfigure [Locations]
  {\includegraphics[width = 0.3\textwidth,trim={1cm 0 4cm 0},clip]{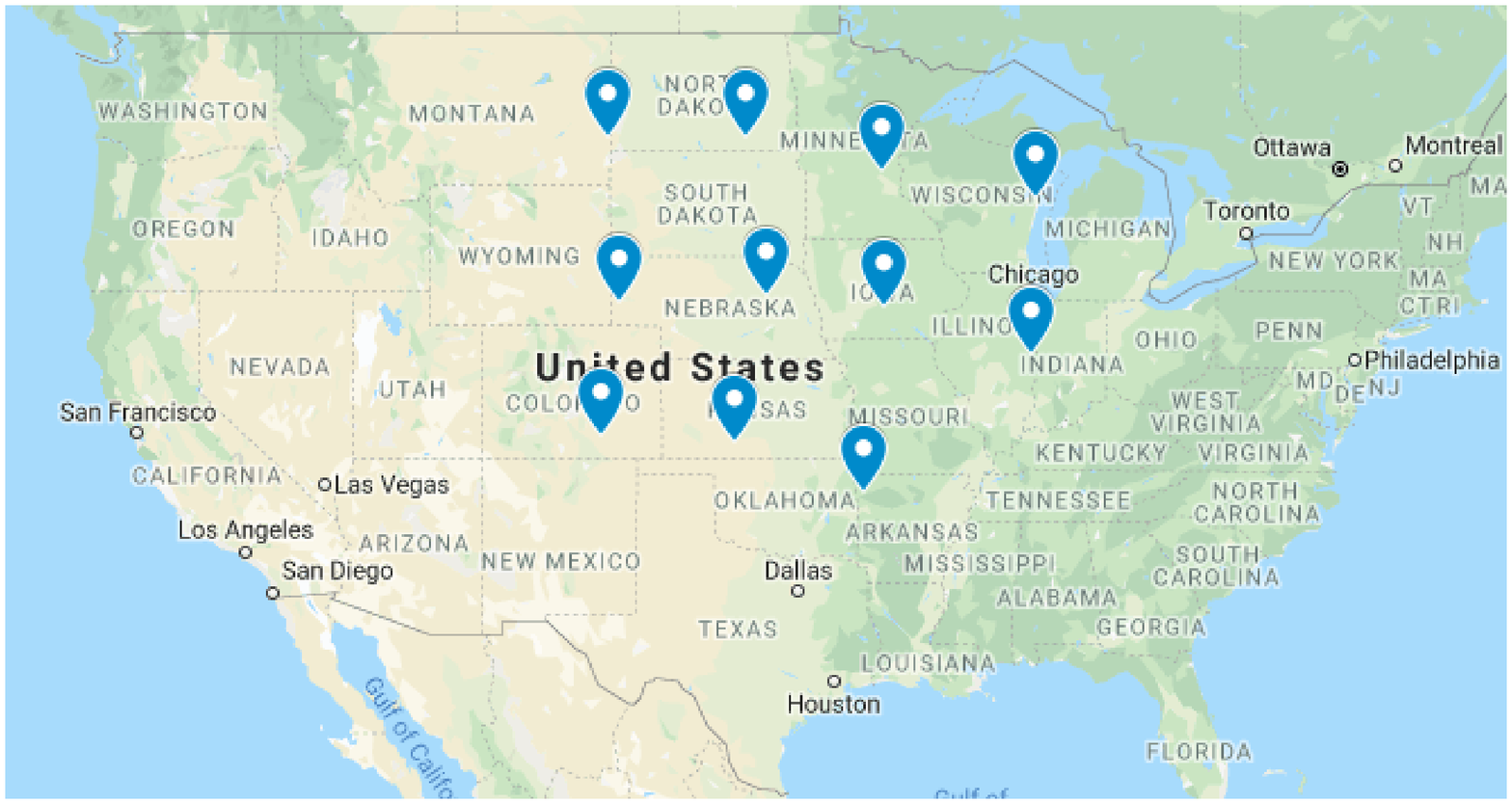} }
  \subfigure [Battery requirement]
  {\includegraphics[width = 0.3\textwidth,trim={1cm 0 0 0}]{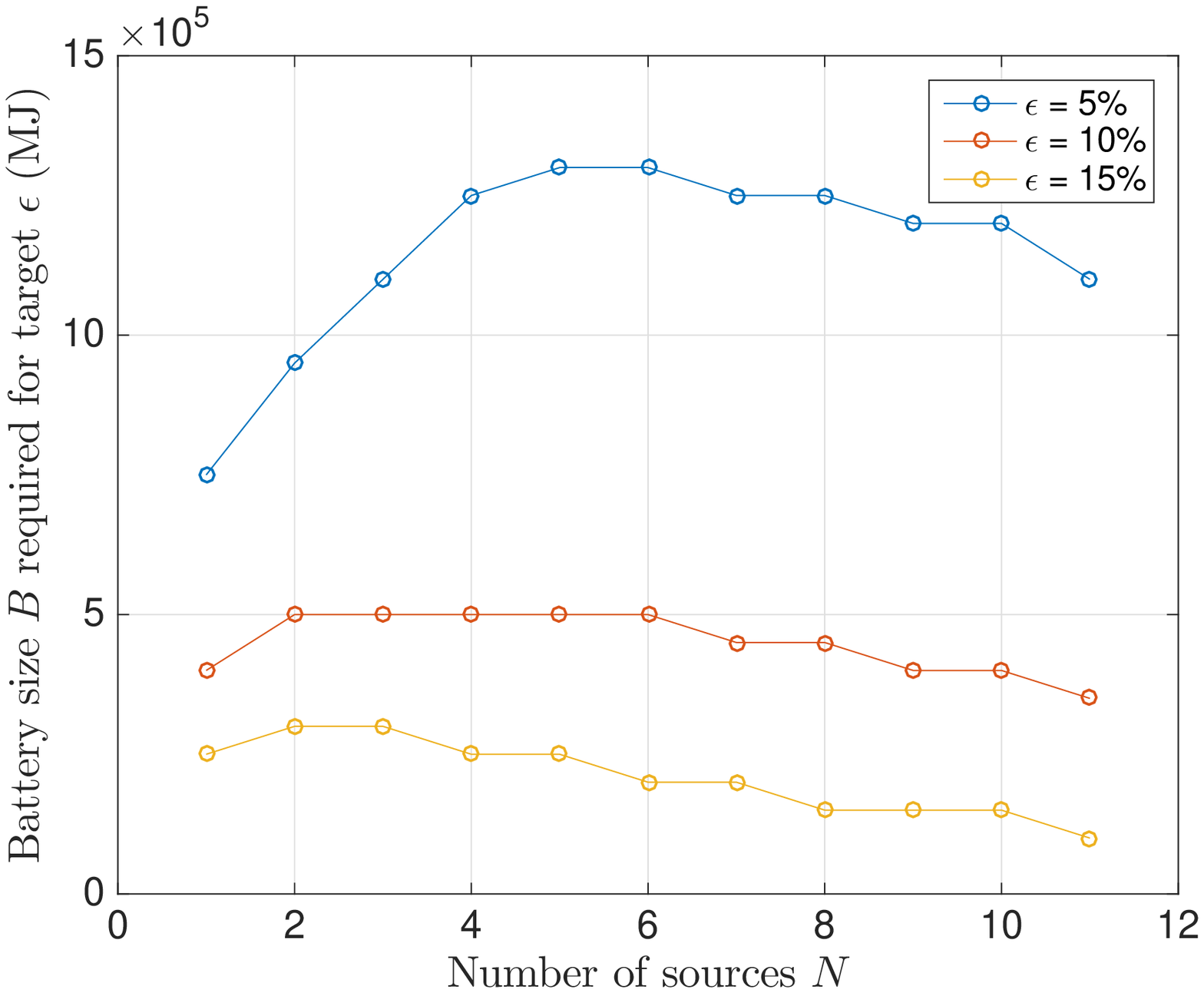} }
  \caption{Scenario 6: Grid with adjacent locations around 500 km apart}
  \label{fig:scenario6}

\end{figure}

To summarize, our case study shows that there are significant
economies of scale to be obtained in practice by sharing a battery
between wind generators even a few hundred kilometers apart. Given the
substantial cost and volume of storage required to `smoothen' the
intermittency in generation at any single location, realizing these
economies of scale would be crucial in order to achieve a high
penetration of renewable generation in the power grid.

\section{Methods: A mathematical foundation for economies of scale in
  battery sharing}
\label{sec:methods}

In the previous section, we made a remarkable observation: When a
battery is shared between $N$ wind generators over 500~km apart, with
each generator contracted to supply 60\% of its average generation to
a prescribed level of reliability, the shared battery size required
does not scale with $N.$ While some economy of scale is to be expected
given the statistical diversity in wind generation across far away
locations, the extent of the observed economy (an $N$-fold reduction
in the volume of storage required) is surprising, and warrants a sound
analytical explanation. 

In this section, we propose a mathematical model with $N$ users,
equipped with their own stochastic net generation (generation minus
demand) process, sharing a common battery of size $B.$ Assuming that
the net generation processes of the users are \emph{statistically
  independent}, we show that the value of $B$ needed to ensure a
prescribed level of reliability is indeed insensitive to $N,$ under a
certain large deviations approximation. This result provides a formal
explanation for our empirically observed economies of scale: Wind
generators that are sufficiently far apart essentially behave as
though their generation processes are statistically independent (from
the standpoint of sizing of a shared battery).

\subsection{Notation}
If $Z(k), k\in \mathbb{Z}$ is a stochastic process with a steady state
distribution $\mu$, we use $Z(\cdot)$ to denote the process and $Z$ to
denote a random variable such that $Z$ has distribution $\mu$. For
sequences $\{x_n\}_{n\in \Nbb}$ and $\{y_n\}_{n \in \Nbb},$ we say
$x_n \sim y_n$ if $\limn \frac{x_n}{y_n}=1.$ For $m \in \mathbb{N}$,
$[m] := \{0,1,2,\cdots,m\}.$

\subsection{Model}
Consider $N$ users denoted $1,\hdots,N$ and let $\{X_i(k)\}_{k \in
  \mathbb{Z}}$ be a background stochastic process associated with user
$i$. At any time $k$, let the \textit{net generation} associated with
user $i$ be a function of the background process, denoted $r_i(X_i(k))
\in \mathbb{Z}.$\footnote{We are assuming here that there is a
  positive granularity with which energy generation/demand are
  measured; this granularity is taken to be 1 without loss of
  generality.} Technically, $r_i(X_i(k)) = g_i(k)-d_i(k),$ where
$g_i(\cdot)$ is the energy generation process and $d_i(\cdot)$ is the
energy demand process of user $i$. We assume that for each $i$,
$X_i(\cdot)$ is an irreducible discrete-time Markov chain (DTMC) over
a finite state space $S_i$ and that the processes $X_1(\cdot),\hdots,
X_N(\cdot)$ are independent. We use $X$ to denote the $N$-tuple of the
background DTMCs,
\[X(k):=(X_1(k),\hdots,X_N(k)),\] and $S=\prod_{i=1}^N S_i$ to denote
the state space of $X(\cdot)$. Let the total generation of the $N$
processes be
\[r(X(k)) := \sum_{i=1}^N r_i(X_i(k)) = \sum_{i=1}^N \left
  (g_i(k)-d_i(k)\right ).\] Let $\pi_i=(\pi_i(s))_{s\in S_i}$ denote
the steady-state distribution of $X_i(\cdot)$. We define the
\textit{drift} $\Delta_i$ associated with user~$i$ as the steady state
average net generation, i.e.,
 \[\Delta_i:= \sum_{s\in S_i}r_i(s) \pi_i(s). \]
 We assume $\Delta_i > 0$ for all $i.$ Clearly, when $\Delta_i>0$, the
 long-run time-averaged generation is greater than the long-run
 time-averaged demand. Define the total drift $\Delta$ as
 \[\Delta := \sum_{i=1}^N \Delta_i = \sum_{s \in S} \pi(s)r(s),\]
 where $\pi$ denotes the stationary distribution of the vector process
 $X(\cdot).$ For technical reasons, we make the following assumptions:
 \begin{enumerate}
 \item[A1.] $\Pbb[X_i(k+1)=s\ |\ X_i(k)=s] > 0$ for all $s \in S_i,$
and    $1 \leq i \leq N.$
 \item[A2.] No user always generates a net energy surplus, i.e.,
   $|\{s \in S_i:\ r_i(s) < 0\}| > 0$ for all $i.$ 
 \end{enumerate}
 Assumption~A1 ensures that each of the DTMCs $X_i(\cdot)$ is
 aperiodic. Note that Assumption~A2 implies that the system on the
 whole does not always enjoy a net energy surplus, i.e., $\quad |\{s
 \in S:\ r(s) < 0\}| > 0.$ Together, Assumptions~A1 and~A2 ensure that
 the battery evolution (described next) is non-trivial.

 Consider a storage battery of capacity $B \in \mathbb{N}$ which is
 charged/discharged according to the net generation process
 $r(X(\cdot))$. The amount of charge in the battery $b(k)$ evolves as
 a deterministic function of $r(X(k))$, regulated between the upper
 cap $B$ and lower cap $0$. Thus, $b(\cdot)$ evolves as
\begin{equation}
\label{eq:bdot}
b(k+1) = \left[b(k) + r(X(k))\right]_{[0,B]},
\end{equation}
where $[z]_{[0,B]} = \min(\max(z,0),B)$ denotes the projection of~$z$
over the interval $[0,B].$ The dynamic above arises from a greedy operation
of the battery subject to the boundary constraints: we charge the
battery using excess generation whenever feasible, and meet the user
demands whenever feasible.
%
It is easy to see that $\{(b(k),X(k))\}$ is a discrete-time Markov
process that evolves over the state space $[B]\times S.$ Under
Assumptions~A1 and~A2, this Markov process has a well defined steady
state distribution.\footnote{Note that certain battery occupancies in
  $[B]$ may not be reachable starting with a full/empty battery. For
  example, if the net generation values $\{r(s):\ s \in S\}$ are all
  even and $B$ is even. However, under Assumptions~A1 and~A2, there is
  a unique recurrent class in $[B]\times S$ over which the steady
  state distribution is supported. This recurrent class includes both
  empty as well as full battery occupancy.}

\subsection{Loss of load probability}
Let $(b,X)$ be a random vector distributed according to the steady
state distribution of $\{(b(k),X(k))\}.$ A \textit{loss of load event}
is an event where at least one user's demand cannot be met. The
\textit{loss of load probability}, or $\LOLP$, is the probability of
this event under the steady state distribution. It is thus the
probability that the battery occupancy is insufficient to meet all
user demand, i.e.,
\begin{equation}
\LOLP = \Pbb[b + r(X) <0].\label{eq:lolpdef} 
\end{equation}

Note that the system $\LOLP$ is the long run fraction of time at least
one user is unable to meet its demand. The loss of load probability
of particular user $i,$ say $\LOLP_i,$ is well defined only once a
scheduling rule is specified for how to allocate the available energy
between demanding users when a loss of load event occurs. However,
under any scheduling policy, $\LOLP_i \leq \LOLP,$ and so ensuring
$\LOLP < \epsilon$ implies $\LOLP_i < \epsilon$.

While the $\LOLP$ can be computed numerically in a routine manner by
solving the system of linear equations that define the stationary
distribution of the DTMC $\{(b(k),X(k))\},$ this does not yield any
insight on the behavior of $\LOLP$ with~$B$ or with $N$. In the
following section, we characterize the asymptotic behavior of $\LOLP$
with increasing $B$. The calculation of the large deviations exponent
of $\LOLP$, and thereby showing its invariance with $N$ is the central
contribution of this paper.

\subsection{Large Deviations Decay Rate Scaling}
\label{sec:results}

In this section, we analyse how the large deviations decay rate of the
$\LOLP$ with respect to the battery size $B$ scales with the number of
users $N.$ This in turn allows us to capture the scaling of the
battery size requirement (in order to meet a given reliability
threshold) with the number of users.

Specifically, we characterize the decay rate $\lambda,$ such that
$\log(\LOLP) \sim -\lambda B$ as $B \ra \infty$ using large deviations
theory. This decay rate characterization implies that a battery size
of approximately $B^{\epsilon} := \frac{1}{\lambda} \log(1/\epsilon)$
would suffice in order to meet an $\LOLP$ target of~$\epsilon.$ Our
main finding is that $\lambda \geq \min_{1 \leq i \leq N} \lambda_i,$
where $\lambda_i$ denotes the decay rate associated with the user~$i$
operating individually (i.e., without battery pooling with the
remaining users and with background process $\{X_i(\cdot)\}$). This
means that $B_{\epsilon} \leq \max_{1 \leq i \leq N} B^{\epsilon}_i,$
where $B^{\epsilon}_i := \frac{1}{\lambda_i} \log(1/\epsilon)$ denotes
the large deviations estimate of the battery size requirement for
user~$i$ operating individually. This insensitivity of $B_{\epsilon}$
with $N$ suggests a tremendous statistical economy of scale that can
be obtained by sharing a common battery between multiple sources.

Recall from the introduction that this economy of scale is contingent
on the assumption that the background processes $\{X_i(\cdot)\},
i=1,\hdots,N$ are independent. It is easy to see from our formulation
that if these processes are instead \textit{identical}, \ie, equal on
all sample paths, the large deviations decay rate will scale inversely
proportional to $N$, thereby offering no economy of scale.

We now introduce some preliminaries required to state the large
deviations decay rate characterization (Theorem~\ref{thm:ldp}).
Define $U_0 = 0,$ $$U_k := \sum_{j=1}^k -r(X(-j)) \quad (k \geq 1).$$
The process $\{\nicefrac{U_k}{k}\}$ satisfies a large deviations
principle (LDP) (see \cite[Section~2.3]{Dembo1998}), with a rate
function that is defined in terms of the following function.
\begin{equation*}
  \label{eq:Lambda}
 \Lambda(\theta):=\lim_{k \ra \infty} \frac{\log \Exp{e^{\theta
       U_k}}}{k}.
\end{equation*}
That $\Lambda(\cdot)$ is well defined, i.e., the limit in the above
definition exists for all $\theta,$ is shown in in
Section~\ref{app:Lambda}.

We are now ready to state the decay rate characterization for the
$\LOLP$ with respect to the battery size $B.$
\begin{theorem}
Let $X_i(\cdot), i=1,\hdots,N$ be independent DTMCs, let the battery
be governed by the process $r(k)\equiv r(X(k))$ in \eqref{eq:bdot},
and define $\LOLP$ as in \eqref{eq:lolpdef}.  If $\Delta_i>0$ for
each~$i$, then,
  \label{thm:ldp}
  $$\lim_{B \ra \infty} \frac{\log \LOLP}{B} = -
  \lambda,$$ where
  \begin{equation}
  \label{eq:decay_rate}
  \lambda := \sup\{\theta > 0\ :\ \Lambda(\theta) < 0\} \in
  (0,\infty).
\end{equation}
\end{theorem}

Theorem~\ref{thm:ldp} states that the $\LOLP$ decays exponentially
with respect to the battery size with decay rate $\lambda.$ That
$\lambda \in (0,\infty)$ follows since $\Lambda(\cdot)$ is a
differentiable convex function with $\Lambda(0) = 0,$
$\Lambda'(0) < 0$ and
$\lim_{\theta \ra \infty} \Lambda(\theta) = \infty$ (see
Section~\ref{app:Lambda}).

Theorem~\ref{thm:ldp} can also applied with $N=1$; this corresponds to
the case where a user $i$, with background process $X_i(\cdot)$, uses
a battery of size $B$ such that the battery dynamics \eqref{eq:bdot}
are governed by $r(k) \equiv r_i(X_i(k))$. If $b_i$ is the resulting
steady state battery occupancy in this case, we have
\begin{equation*}
\lim_{B \rightarrow \infty} \frac{\log \Pbb[b_i + r_i(X_i) < 0]}{B} = -
\lambda_i, \label{eq:lambdai}
\end{equation*}
where $\lambda_i := \sup\{\theta > 0\ :\ \Lambda_i(\theta) < 0\},$ and
$$\Lambda_i(\theta):=\lim_{k \ra \infty} \frac{\log \Exp{e^{-\theta
      \sum_{j=1}^k r_i(X_i(-j))} }}{k}.$$ Indeed, Theorem~\ref{thm:ldp}
allows us to compare the decay rate $\lambda$ for the collective
battery operation with the decay rates $\lambda_i$ associated with
standlone operations by the individual users, as is shown in
Corollary~\ref{corr:decay_rate_scaling}.

\begin{corollary} Let the setting of Theorem~\ref{thm:ldp} hold. Then 
\label{corr:decay_rate_scaling}
  $$\lambda \geq \min_{1 \leq i \leq n} \lambda_i.$$
\end{corollary}
\begin{proof}
Due to the independence between the background processes $X_i(\cdot),$
it follows that
$$\Lambda(\theta) = \sum_{i=1}^N \Lambda_i(\theta).$$ Defining
$\lambda_{\min} = \min_{1 \leq i \leq n} \lambda_i,$ the statement of
the corollary follows from the observation that
$\Lambda(\lambda_{\min}) \leq 0.$ 
\end{proof}

To interpret Corollary~\ref{corr:decay_rate_scaling}, note that a
smaller value of decay rate implies that the $\LOLP$ decays more
slowly with $B,$ which in turn implies that a larger battery is
required in order to meet a given reliability target. Indeed, the
large deviations estimate of the battery size required to meet a
reliability target of $\epsilon$ is
$\frac{1}{\lambda} \log(1/\epsilon).$
Corollary~\ref{corr:decay_rate_scaling} states that the $\LOLP$ of the
combined system decays at least as fast as the slowest decaying
$\LOLP$ for the $N$ users operating alone. Thus, the battery
requirement for the combined system is at most the `worst case' among
the battery requirements for each of the $N$ users operating alone. Indeed, note that $\lambda > \lambda_{\min}$ if $\lambda_i \neq \lambda_j$ for some $i \neq
j.$

Of course, the above scaling for the battery size requirement rests on
two key approximations:
\begin{enumerate}
\item We work with large buffer asymptotics, which are provably
  accurate only in the limit $\epsilon \da 0.$
\item We provide \emph{logarithmic} rather than \emph{exact} $\LOLP$
  asymptotics. Exact asymptotics of the form $\LOLP \sim C_N
  e^{-\lambda B}$ would more accurately characterize the sensitivity
  of the battery size requirement with $N.$
\end{enumerate}
Finally, we note that the scaling in Theorem~\ref{thm:ldp} is
different from the \emph{many sources scaling} that is common in the
networking literature (see \cite{BigQueues}), where the buffer content
$B$ as well as the number of sources $N$ are scaled proportionately to
infinity. In Theorem~\ref{thm:ldp}, we perform a \emph{large buffer
scaling}, where the number of sources $N$ is held fixed, while the
buffer size $B$ is scaled to infinity.

In the remainder of this section, we present the proof of
Theorem~\ref{thm:ldp}.

\subsection{Proof of Theorem~\ref{thm:ldp}}

Our first observation is that the $\LOLP$ decay rate matches that
associated with the long run fraction of time the battery is empty,
i.e.,
\begin{equation}
  \label{eq:lolp_batt_empty}
  \lim_{B \ra \infty} \frac{\log \prob{b + r(X) < 0}}{B} = \lim_{B \ra
    \infty} \frac{\log \prob{b = 0}}{B}.
\end{equation}
We provide the proof of \eqref{eq:lolp_batt_empty} in
Section~\ref{app:lolp_batt_empty}.

Next, we analyse the large buffer asymptotics of $\prob{b = 0}$ via
the \emph{reversed} system \cite{Mitra88}, which is obtained by
interchanging the role of generation and
demand. Thus, $$r^{\rev}(X(k)) = -r(X(k)),$$ where we use the
superscript $\rev$ to represent quantities in the reversed
system. Moreover, $\Delta^{\rev} = -\Delta.$ Since the original system
is associated with a positive drift ($\Delta > 0$), the reversed
system is associated with negative drift ($\Delta^{\rev} < 0$). It not
hard to see that
\begin{equation}
  \label{eq:battery_rev}
  \prob{b = 0} = \prob{b^{\rev} = B}
\end{equation}
i.e., the long run fraction of time the battery is empty in the
original system equals the long run fraction of time the battery is
full in the reversed system (this observation seems to have been made
first in \cite{Mitra88}).

In light of \eqref{eq:lolp_batt_empty} and \eqref{eq:battery_rev}, it
suffices to show that
\begin{equation}
  \label{eq:decay_rate_rev}
  \lim_{B \ra \infty} \frac{\log \prob{b^{\rev} = B}}{B} = -\lambda,
\end{equation} where $\lambda$ is defined by
\eqref{eq:decay_rate}. Finally, \eqref{eq:decay_rate_rev} follows by
noting that in the reversed system, the buffer evolution is given
by $$b^{\rev}(k+1) = \left[b^{\rev}(k) + Y(k)\right]_{[0,B]},$$ where
$Y(k) = -r(X(k)).$
This correponds to a finite buffer queue with an increment process
$Y(\cdot).$ Logarithmic large buffer asymptotics of the form
\eqref{eq:decay_rate_rev} are known for such systems for a broad class
of stationary increment processes, which includes the present case,
where $Y(\cdot)$ is a function of the background Markov process
$X(\cdot);$ see \cite{Toomey98} and \cite[Section~6.5]{BigQueues}.

That $\lambda \in (0,\infty)$ follows from the characterization of the
function $\Lambda(\cdot)$ in the present setting, which is done below
in Section~\ref{app:Lambda}.


\subsection{Proof of~\eqref{eq:lolp_batt_empty}}
\label{app:lolp_batt_empty}

First, we note that $b(k) + r(X(k)) < 0$ implies $b(k+1) = 0.$ It
therefore follows that $$\prob{b + r(X)} < 0 \leq \prob{b = 0}.$$ It
therefore suffices to show that there exists a positive constant~$c$
such that
\begin{equation}
\label{eq:lolp_batt_empty_1}
\prob{b + r(X) < 0} \geq c \prob{b = 0}.
\end{equation}
This follows from our assumption A1, which ensures whenever $b(k) =
0,$ there is a positive probability of a loss of load at time $k+1$
(i.e., $b(k)+r(X(k)) < 0$). The formal argument is as follows.

Define $S_- := \{s \in S:\ r(s) \leq 0.\}$ By Assumption A2, $|S_-| >
0.$ Let $\tilde{X}(k) := X(k-1).$ Note that $(b(k),\tilde{X}(k))$ is
an ergodic discrete-time Markov chain. Moreover, $$\prob{b = 0}
= \sum_{s \in S_-} \prob{b=0,\tilde{X} = s}.$$

Now, pick $s \in S_-.$ Define the renewal process defined by the
renewal instants $(k:\ b(k)=0,\tilde{X}(k)=s.)$ Let the reward $R_j$
over renewal cycle~$j$ be the number of loss of load events in that
cycle. By the renewal reward theorem
(see \cite[Section~5.4]{Gallager2013}),
\begin{align}
\prob{b + r(X) < 0} &= \Exp{R_j} \prob{b = 0,\tilde{X}=s} \nn \\
&\geq p \prob{b = 0,\tilde{X}=s}, \label{eq:eq:lolp_batt_empty_2}
\end{align}
where $p = \min_{s\in S_-} \prob{\tilde{X}(k+1)=s|\ \tilde{X}(k)=s} >
0$ (by Assumption A1). Summing \eqref{eq:eq:lolp_batt_empty_2} over
$s \in S_-,$ we see that $$\prob{b + r(X) <
0} \geq \frac{p}{|S_-|} \prob{b = 0}.$$ This completes the proof.

\subsection{Characterization of $\Lambda(\cdot)$}
\label{app:Lambda}

Define $X^*(k):= X(-k)$ denote the time-reversal of the process
$X(\cdot).$ Note that $X^*$ is a DTMC with transition probability
matrix $P^*,$ where $P^*_{s,s'} = \nicefrac{\pi(s')
P_{s',s}}{\pi(s)}.$ Here, $P$ denotes the transition probability
matrix associated with $X(\cdot).$ It now follows from an application
of the Perron Frobenius theorem that $\Lambda(\theta)
= \log \rho(M(\theta)),$ where $\rho(M(\theta))$ is the Perron
Frobenuis eigenvalue of the non-negative matrix $M(\theta)$ defined
as $$M_{s,s'}=P^*_{s,s'}e^{-\theta r(s)}$$
(see \cite[Example~2.15]{BigQueues}).

The function $\Lambda(\cdot)$ has the following properties that are
relevant to us.
\begin{enumerate}
\item $\Lambda(\cdot)$ is convex and differentiable
\item $\Lambda'(0) = \sum_s -\pi(s)r(s) = -\Delta < 0$
\item $\lim_{\theta \ra \infty} \Lambda(\theta) = \infty$
\end{enumerate}

That $\Lambda(\cdot)$ is convex follows from the fact that it is
pointwise limit of convex functions. Its differentiability follows
from the differentiability of Perron Frobenius eigenvalue of a
non-negative matrix with respect to its entries. Statement~(2) above
follows from Lemma~3.2 in~\cite{BigQueues}.

Finally, to show Statement~(3), we use
$\rho(M(\theta)) \geq \rho(G(\theta)),$ where $G(\theta)$ is a
symmetric non-negative matrix defined as $$G_{s,s'}(\theta)
= \sqrt{M_{s,s'}(\theta)M_{s',s}(\theta)};$$ (see Theorem~2
in \cite{Schwenk1986}). It therefore suffices to show that
$\rho(G(\theta)) \ra \infty$ as $\theta \ra \infty,$ which follows
trivially from the observation that
$\mathrm{trace}(G(\theta)) \ra \infty$ as $\theta \ra \infty$ (note
that all diagonal entries of $G$ are positive due to Assumption~A1).

\section{Concluding Remarks}

Our results motivate several important avenues for future work. On the
analytical front, sharper large buffer asymptotics (say exact rather
than logarithmic) would give a finer characterization of the
dependence of battery requirement on the number of users. Also, it
would be very interesting to model weakly dependent users to obtain a
battery scaling that is intermediate between the $O(1)$ scaling with
independent users and the $O(N)$ scaling that results when the net
generation is perfectly correlated across users.

But more importantly, it is important to note that the battery sharing
model we consider assumes that there are no transmission contraints
that limit the charging/discharging of the battery by the different
users/locations. In practice, battery sharing between independent
users (say wind generators that are geographically far apart) will
likely be affected by line constraints (as well as regulatory
constraints). This work thus motivates a deeper understanding of the
tradeoff between battery sizing, battery placement, and the
provisioning of transmission capacity.

Moreover, enabling battery sharing between geographically distributed
wind generators would also require pricing and regulatory
innovations. Several policy questions would need to be
addressed. Should these generators be allowed to participate and bid
in electricity markets as one combined entity? How should this entity
compensate the grid for the transmission infrastructure required for
battery sharing? Should such conglomerations be permitted to form
across ISO jurisdictions?

\ignore{ We considered a
  battery of size $B$ shared between $N$ users having independent net
  generation processes with positive steady state mean. Our main
  thrust in this paper is showing that if this battery suffices for
  each user individually to achieve a $\LOLP \leq \epsilon$, it
  suffices for all of them collectively.  We found that as we increase
  the battery size $B$, the $\LOLP$ decays (asymptotically) at least
  as fast as the slowest individual decay rate. We showed this using
  large deviations techniques.  We numerically validated this result
  for synthetic data, and also from real data.  }

\bibliographystyle{IEEEtran}
\bibliography{references}

\end{document}